\pdfoutput=1
\documentclass[11pt]{article}

\usepackage[T1]{fontenc}
\usepackage[utf8]{inputenc}
\usepackage{lmodern}
\usepackage[margin=1in]{geometry}
\usepackage{amsmath,amssymb,amsthm,mathtools}
\usepackage{microtype}
\usepackage[hidelinks]{hyperref}

\input{glyphtounicode}
\hypersetup{
  pdftitle={The Born Rule for Projective Measurements from Metric Non-Expansion and Calibration},
  pdfauthor={Aaron Lax},
  pdfsubject={Born-rule rigidity for calibrated projective readouts and its finite-POVM boundary},
  pdfkeywords={Born rule, Fisher information, Fubini-Study metric, operational calibration, projective measurements, POVM boundary, classical post-processing, effect noncontextuality}
}

\newtheorem{lemma}{Lemma}
\newtheorem{theorem}{Theorem}
\newtheorem{proposition}{Proposition}
\newtheorem{corollary}{Corollary}
\theoremstyle{remark}
\newtheorem{remark}{Remark}

\newcommand{\C}{\mathbb{C}}
\newcommand{\R}{\mathbb{R}}
\newcommand{\CP}{\mathbb{CP}}
\newcommand{\dd}{\mathrm{d}}
\newcommand{\id}{\operatorname{id}}
\newcommand{\Var}{\operatorname{Var}}
\newcommand{\ket}[1]{\lvert #1\rangle}
\newcommand{\bra}[1]{\langle #1\rvert}
\newcommand{\braket}[2]{\langle #1 \mid #2\rangle}
\newcommand{\norm}[1]{\left\lVert #1 \right\rVert}
\newcommand{\inner}[2]{\left\langle #1,#2 \right\rangle}

\title{The Born Rule for Projective Measurements\\
       from Metric Non-Expansion and Calibration}
\author{Aaron Lax\\{\small Independent Researcher}}
\date{July 2026}

\begin{document}
\maketitle

\begin{abstract}
Why should a calibrated quantum measurement report Born probabilities?
Fix a finite dimension $d\geq2$ and a projective measurement $M$ on $\C^d$, and let $P_M$ be a map assigning outcome probabilities to pure states.
I prove that three per-apparatus hypotheses force $P_M$ to be the Born readout:
the square-root readout $R_M=\sqrt{P_M}$ is absolutely continuous along Fubini--Study geodesics;
the classical Fisher speed of the output never exceeds the quantum Fisher speed of the input almost everywhere on smooth pure-state curves;
and the readout reports certainty on every state in each labeled eigenspace.
Taking square roots turns probabilities into coordinates on a spherical orthant, and calibration pins the vertices.
The metric hypotheses make $R_M$ globally $1$-Lipschitz from Fubini--Study to round distance, so the readout cannot move farther from any calibrated vertex: every coordinate is at least its Born value, both vectors have unit norm, and they are equal.
The hypotheses are per measurement; imposing them in every projective context yields the Born assignment on all projective measurements, at every $d\geq2$, with noncontextuality arriving as an output rather than an axiom.
The metric argument stops at projective measurements: for every non-projective POVM, an explicit family of non-Born readouts satisfies the same metric hypotheses even with the stated certainty-of-occurrence calibration.
\end{abstract}

\section{Introduction}

Why should a calibrated quantum measurement obey the Born rule?
Gleason's theorem answers by constraining probability assignments across every orthogonal resolution of the identity at once, and it fails at $d=2$ \cite{Gleason1957}; envariance and decision-theoretic routes add symmetry or rationality principles on top of the formalism \cite{Zurek2005,Wallace2009}.
This paper asks a narrower question.
Fix one projective measurement, calibrated so that its outcome labels mean what they say.
How much freedom does its pure-state readout retain?

The answer is none, given two metric conditions.
Suppose the square-root readout $R_M=\sqrt{P_M}$ is absolutely continuous along Fubini--Study geodesics, and suppose the readout never converts projective distinguishability into more classical distinguishability: its classical Fisher speed is bounded by the quantum Fisher speed almost everywhere along every smooth curve of pure states.
Then the readout is the Born readout, in every finite dimension $d\geq2$ and for projective measurements of arbitrary rank (Theorems~\ref{thm:pvm-readout} and~\ref{thm:rank-general}).

The proof has one picture.
Taking square roots turns probability vectors into points of a spherical orthant, and calibration pins the orthant's vertices.
The two metric hypotheses make the readout $1$-Lipschitz from Fubini--Study distance to round distance, and a non-expanding map cannot pull any point farther from a fixed vertex.
Every coordinate of the readout is therefore at least its Born value, and both vectors have unit norm, so they are equal.

All hypotheses are imposed one measurement at a time.
Imposing them in every projective context yields the Born assignment on all projective measurements, and agreement across contexts comes out as a consequence rather than an assumption (Corollary~\ref{cor:all-contexts}).

The argument works for projective measurements and stops there.
For every non-projective POVM, an explicit family of non-Born readouts satisfies the same metric hypotheses even with certainty-of-occurrence calibration on every outcome (Theorem~\ref{thm:povm-boundary}): unsharp outcomes do not supply enough certainty states to anchor the map.
Section~\ref{sec:povm-closure} records two conditional repairs and names the identification principle each one spends.

\section{Setup}\label{sec:setup}

Fix $d \geq 2$ and an orthonormal basis $\{\ket{e_1},\ldots,\ket{e_d}\}$ of $\C^d$.
Write $\CP^{d-1}$ for projective pure-state space and
\[
\Delta^{d-1}=\left\{u\in\R^d: u_i\geq 0,\ \sum_i u_i = 1\right\}
\]
for the closed probability simplex. Let
\[
M=\{\ket{e_1}\bra{e_1},\ldots,\ket{e_d}\bra{e_d}\}
\]
be the fixed rank-1 PVM in this basis, and let
\[
P_M:\CP^{d-1}\to\Delta^{d-1}
\]
be a probability readout map for $M$. Define its square-root readout
\[
R_M([\psi])_i = \sqrt{P_M([\psi])_i},
\]
so
\[
R_M:\CP^{d-1}\to S^{d-1}_+,
\qquad
S^{d-1}_+ = \left\{x\in\R^d: x_i\geq 0,\ \sum_i x_i^2=1\right\}.
\]

For a curve $\gamma:I\to\CP^{d-1}$, write $r=R_M\circ\gamma$ and define its upper metric speed and the associated extended readout Fisher speed by
\begin{equation}\label{eq:upper-speed}
 \operatorname{md}^{+}r(s)
 =\limsup_{\substack{h\to 0\\ s+h\in I}}
 \frac{d_{\mathrm{round}}(r(s+h),r(s))}{|h|},
 \qquad
 F_{\mathrm{cl}}(P_M\circ\gamma;s)
 =4\bigl(\operatorname{md}^{+}r(s)\bigr)^2.
\end{equation}
The value in \eqref{eq:upper-speed} may be infinite and does not presuppose absolute continuity of $r$.
If $r$ is absolutely continuous, its ordinary metric derivative $\operatorname{md}r$ exists almost everywhere, agrees there with $\operatorname{md}^{+}r$, and satisfies
\begin{equation}\label{eq:length-metric-derivative}
 L_{\mathrm{round}}(r)=\int_I \operatorname{md}r(s)\,\dd s;
\end{equation}
see Ambrosio--Gigli--Savar\'e \cite[Sec.~1.1]{AmbrosioGigliSavare2008}.
For a smooth input curve, write
\[
F_Q(\gamma;s)=4g_{\mathrm{FS},\gamma(s)}(\dot\gamma(s),\dot\gamma(s)).
\]

I use three hypotheses.

\medskip
\noindent\textbf{(H1) Square-root regularity.}
$R_M$ is continuous on $\CP^{d-1}$ and absolutely continuous along affinely parameterized Fubini--Study geodesic segments.

\medskip
\noindent\textbf{(H2) Universal Fisher non-expansion.}
For every smooth pure-state curve $\gamma:I\to\CP^{d-1}$,
\[
F_{\mathrm{cl}}(P_M\circ\gamma;s) \leq F_Q(\gamma;s)
\qquad\text{for almost every }s\in I.
\]

\medskip
\noindent\textbf{(H3) Operational calibration.}
For the basis preparations,
\[
P_M([e_i]) = \delta_i, \qquad i=1,\ldots,d,
\]
where $\delta_i$ is the $i$th simplex vertex.

The primary scope is one dimension $d$, one fixed rank-1 PVM $M$, and pure states only; Section~\ref{sec:rank-general} removes the rank restriction.
Every hypothesis is per measurement, and no cross-context axiom enters before the conditional POVM extensions of Section~\ref{sec:povm-closure}, on which the projective results do not depend.

For later use, recall two standard metric identities. On the open simplex
\[
\Delta^{\circ}_{d-1}=\{u\in\Delta^{d-1}:u_i>0\text{ for all }i\},
\]
the Fisher metric is
\begin{equation}\label{eq:fisher-metric}
 g^F_u(v,v)=\sum_i \frac{v_i^2}{u_i},\qquad \sum_i v_i=0.
\end{equation}
On pure states,
\begin{equation}\label{eq:fq-fs}
 F_Q=4g_{\mathrm{FS}},
\end{equation}
where $g_{\mathrm{FS}}$ is the Fubini--Study metric \cite{Wootters1981,BraunsteinCaves1994}.
If $p_i=P_M([\psi])_i$ and $R_{M,i}=\sqrt{p_i}$, then at points where the square-root output is differentiable,
\begin{equation}\label{eq:fcl-sqrt}
 F_{\mathrm{cl}}(p;s)=4\left\|\frac{\dd R_M}{\dd s}\right\|_{\mathrm{round}}^2
 =\sum_i \frac{(\dd p_i/\dd s)^2}{p_i},
\end{equation}
where $\norm{\cdot}_{\mathrm{round}}$ on a tangent vector denotes the ambient Euclidean norm restricted to the sphere's tangent space, i.e.\ the round Riemannian norm, and the final expression applies on the open positive part.
Equation \eqref{eq:fcl-sqrt} is the square-root form of classical Fisher information. It is the classical Fisher--Rao/Hellinger identity on the simplex; see Rao \cite{Rao1945} and Amari--Nagaoka \cite{AmariNagaoka2000}.
The displayed elementary formula $F_{\mathrm{cl}}(p)=\sum_i(\dd p_i/\dd s)^2/p_i$ is under-specified at boundary points where some $p_i$ vanishes;
throughout this paper the defining expression is the upper-metric-speed convention \eqref{eq:upper-speed}, which agrees almost everywhere with the ordinary square-root velocity on absolutely continuous output curves and recovers the elementary formula on the open positive part.
On the quantum-geometric side, the square-root/spherical chart used here is the projective-Hilbert geometry developed by Brody and Hughston \cite{BrodyHughston1998} as a Riemannian formulation of quantum mechanics on $\CP^{d-1}$.

\section{Status of the Hypotheses}\label{sec:hypotheses}

Both (H2) and (H3) are admissibility postulates on candidate readouts, not theorems inherited from quantum mechanics with Born already in place.
Each does distinct work in the proof, and each can be defended on its own geometric/operational grounds. I make their status explicit before using them.

\medskip
\noindent\textbf{(H2) and Braunstein--Caves.}
Braunstein and Caves \cite{BraunsteinCaves1994} prove $F_{\mathrm{cl}} \leq F_Q$ for measurement statistics generated by quantum POVMs with Born readout, optimized over POVMs, with saturation at the symmetric logarithmic derivative measurement.
Their inequality presupposes Born to generate the outcome distributions whose Fisher information is being bounded.
Hypothesis (H2) here applies the same inequality to candidate non-Born readouts $P_M$.
This is a new admissibility postulate, not a corollary of Braunstein--Caves.

The geometric defense of (H2) rests on the identity $F_Q\equiv 4g_{\mathrm{FS}}$ on pure states.
The right-hand side is fixed by the projective Fubini--Study geometry on $\CP^{d-1}$, independent of which readout one is contemplating.
Throughout, I take the Fubini--Study metric as primitive ray geometry, induced by the Hermitian inner product on the manifold of rays \cite{ProvostVallee1980,BrodyHughston1998}; I do not justify it operationally through statistical distinguishability.
Wootters' identification of statistical distance with the Fubini--Study angle presupposes Born-form outcome probabilities on the quantum side, so it cannot serve as an independent operational justification inside a Born derivation; I cite \cite{Wootters1981} only for the value of the angle identity.
$F_Q$ provides the admissibility benchmark imposed by (H2): it measures the intrinsic distinguishability of the state family in the underlying projective metric, and a readout exceeding it manufactures distinguishability the geometry does not supply, as the quartic escort below does.
Hypothesis (H2) says only that a candidate readout may discard distinguishability but cannot create more than the geometry supplies.
This is the natural admissibility condition for any information-extraction map, and it does not single out Born by itself.

\medskip
\noindent\textbf{Born membership and the phase-variance decomposition.}
The Born readout itself satisfies (H1)--(H3), and its exact relation to the bound in (H2) is an identity worth displaying: the deficit $F_Q-F_{\mathrm{cl}}$ is four times the variance of the phase velocities that the fixed basis cannot see.

\begin{lemma}[Phase-variance decomposition]\label{lem:phase-variance}
Let $B_M([\psi])_i=|\braket{e_i}{\psi}|^2$ be the Born readout for the fixed PVM $M$.
\begin{enumerate}
\item $B_M$ satisfies (H1), (H2), and (H3).
\item Let $[\psi(s)]$ be a smooth pure-state curve whose Born probabilities satisfy $p_i(s)>0$ for all $i$, and write $\braket{e_i}{\psi(s)}=\sqrt{p_i(s)}\,e^{i\varphi_i(s)}$ for a smooth normalized representative. Then
\begin{equation}\label{eq:phase-variance}
F_Q = F_{\mathrm{cl}}(p) + 4\Var_p(\dot\varphi),
\qquad
\Var_p(\dot\varphi)=\sum_i p_i\dot\varphi_i^2-\Bigl(\sum_i p_i\dot\varphi_i\Bigr)^{\!2}.
\end{equation}
In particular $F_{\mathrm{cl}}=F_Q$ along such a curve if and only if the phase velocities $\dot\varphi_i$ agree for all $i$, i.e.\ exactly when, up to a global phase, the phases $\varphi_i$ are constant in $s$.
\end{enumerate}
\end{lemma}

\begin{proof}
(H3) is immediate, and (H1) holds because each $|\braket{e_i}{\psi(s)}|$ is locally Lipschitz in the parameter along smooth curves, hence absolutely continuous on compact parameter intervals.
For (H2), choose on any parameter interval a horizontal representative, $\braket{\psi}{\dot\psi}=0$, and set $z_i=\braket{e_i}{\psi}$.
Then $F_Q=4\braket{\dot\psi}{\dot\psi}=4\sum_i|\dot z_i|^2$.
The curve $(|z_i|)_i$ is absolutely continuous, so \eqref{eq:upper-speed} agrees almost everywhere with its ordinary round velocity; the inequalities $|\dd|z_i|/\dd s|\leq|\dot z_i|$ therefore give $F_{\mathrm{cl}}\leq F_Q$ almost everywhere.

For \eqref{eq:phase-variance}, differentiate $\psi(s)=\sum_i\sqrt{p_i(s)}\,e^{i\varphi_i(s)}\ket{e_i}$:
\[
\braket{\dot\psi}{\dot\psi}=\sum_i\Bigl(\frac{\dot p_i^2}{4p_i}+p_i\dot\varphi_i^2\Bigr),
\qquad
\braket{\psi}{\dot\psi}=\tfrac12\sum_i\dot p_i+i\sum_i p_i\dot\varphi_i=i\sum_i p_i\dot\varphi_i,
\]
using $\sum_i\dot p_i=0$. Hence
\[
F_Q=4\bigl(\braket{\dot\psi}{\dot\psi}-|\braket{\psi}{\dot\psi}|^2\bigr)
=\sum_i\frac{\dot p_i^2}{p_i}+4\Var_p(\dot\varphi)
=F_{\mathrm{cl}}(p)+4\Var_p(\dot\varphi).
\]
The variance vanishes if and only if $\dot\varphi_i$ is independent of $i$ on the full support, which is a global phase reparametrization.
\end{proof}

Consequently every interior simplex point and tangent pair $(p,v)$ is realized at $s=0$ by a constant-phase curve $\psi(s)=\sum_i\sqrt{p_i(s)}\,\ket{e_i}$ with $p(0)=p$, $\dot p(0)=v$, along which Born saturates $F_{\mathrm{cl}}=F_Q$.
This is the transfer device used for the escort corollary in Section~\ref{sec:corollaries}.

\medskip
\noindent\textbf{Metric, not estimator, use of Fisher.}
At differentiability points in the open simplex, classical Fisher information is the Fisher--Rao metric evaluated on tangent vectors, equivalently $4$ times the squared round velocity in the square-root chart.
Equation \eqref{eq:upper-speed} extends that square-root speed through the simplex boundary and to arbitrary output curves without assuming (H1).
Thus (H2) is not a claim about any finite-sample estimator of Fisher information, nor about reconstruction from sampled trajectories.
It is an almost-everywhere differential non-expansion condition between the projective Fubini--Study metric and the classical Fisher--Rao/Hellinger metric.
By \eqref{eq:fcl-sqrt} the classical Fisher--Rao metric is four times the round (Hellinger/Bhattacharyya) metric in the square-root chart, and $F_Q=4g_{\mathrm{FS}}$ by \eqref{eq:fq-fs}; the two factors of four cancel, so (H2) is equivalently an almost-everywhere \emph{Fubini--Study to Hellinger speed non-expansion} condition.
I use the standard normalizations on both sides.
On the classical side, by Cencov's theorem, extended by Campbell, the Fisher--Rao metric is the unique Riemannian metric on probability simplices, up to one global scale, that is monotone under congruent Markov embeddings \cite{Cencov1981,Campbell1986}; monotonicity therefore fixes the classical metric only up to scale, not its coefficient.
The scale in (H2) is fixed by adopting the angle conventions Fubini--Study distance $\arccos|\braket{\phi}{\psi}|$ and Hellinger (Bhattacharyya) distance $\arccos\sum_i\sqrt{p_iq_i}$: both take values in $[0,\pi/2]$, both vanish exactly on identical pairs, and both attain $\pi/2$ exactly on perfectly distinguishable pairs: orthogonal states, disjoint supports.
Under these standard round-sphere conventions, no adjustable constant appears in (H2); the matched coefficient is a normalization choice, not a consequence of Cencov's uniqueness theorem \cite{Bhattacharyya1943}.
The readout may discard local projective distinguishability but may not create greater statistical speed; (H1) is what upgrades that local bound to the distance bound used below.
Questions of sampling bias or estimator consistency enter only when one tries to estimate these geometric quantities from data; they are external to the rigidity theorem below.

\medskip
\noindent\textbf{(H1) lives in the square-root chart.}
Regularity is imposed on $R_M=\sqrt{P_M}$ rather than on $P_M$ because absolute continuity does not pass through the square root at the simplex boundary.
Take $d=2$ and the probability curve $p_1(s)=s^2\sin^2(1/s)$ for $s\neq 0$, $p_1(0)=0$, $p_2=1-p_1$: the derivative $2s\sin^2(1/s)-\sin(2/s)$ is bounded, so $p$ is Lipschitz, hence absolutely continuous, while $\sqrt{p_1(s)}=|s\sin(1/s)|$ has unbounded variation on every neighborhood of $s=0$ and is not absolutely continuous.
The converse direction is automatic: if $R_M$ is absolutely continuous along a geodesic segment then so is $P_M=R_M^2$, since products of bounded absolutely continuous functions are absolutely continuous.
So stating (H1) in the square-root chart is the strictly stronger choice, and it is the coherent one: by \eqref{eq:fcl-sqrt} the Fisher--Rao geometry \emph{is} the round geometry of the square-root variables, so the chart that carries the metric hypothesis (H2) must carry the regularity hypothesis as well.

\medskip
\noindent\textbf{(H3) is logically weaker than the full Born formula.}
(H3) fixes the readout on the $d$ basis preparations only. All rays outside the calibrated set $\{[e_1],\ldots,[e_d]\}$ are unconstrained by (H3) alone.
(H1) gives the local metric condition global force, (H2) supplies the geometric constraint on smooth input curves, (H3) anchors the readout at the $d$ simplex vertices, and Theorem~\ref{thm:pvm-readout} forces their joint consequence to be Born.

\medskip
\noindent\textbf{No hypothesis follows from the other two.}
The separation of the hypotheses is visible in four simple readouts.
Each hypothesis fails for a readout satisfying the other two.
For (H1), the staircase below shows specifically that absolute continuity, which permits integration of the almost-everywhere speed bound, cannot be weakened to continuity.

\smallskip
\noindent\emph{Uniform readout (fails (H3) only).}
The constant readout $P_M([\psi])=(1/d,\ldots,1/d)$ satisfies (H1), and satisfies (H2) trivially: $F_{\mathrm{cl}}=0\leq F_Q$ on every smooth curve, since the readout discards all distinguishability.
It fails (H3): basis preparations $\ket{e_i}$ are not reported with certainty.

\smallskip
\noindent\emph{Permuted Born (fails (H3) only).}
For any non-trivial permutation $\sigma$ of $\{1,\ldots,d\}$, set
$P_M([\psi])_i=|\braket{e_{\sigma(i)}}{\psi}|^2$.
This readout is Born up to relabeling, and coordinate permutations are isometries of the round metric, so by Lemma~\ref{lem:phase-variance} it satisfies (H1) and (H2) exactly as Born does: $F_{\mathrm{cl}}\leq F_Q$ along every smooth pure-state curve, with equality on full-support curves precisely when the phases can be taken constant \eqref{eq:phase-variance}.
But for the labeled PVM $M=\{\ket{e_i}\bra{e_i}\}$ in its declared labeling,
$P_M([e_i])=\delta_{\sigma^{-1}(i)}\neq \delta_i$ whenever $\sigma(i)\neq i$, so it fails (H3).

\smallskip
\noindent\emph{Quartic escort (fails (H2) only).}
Set
\[
P_M([\psi])_i=\frac{|\braket{e_i}{\psi}|^4}{\sum_j|\braket{e_j}{\psi}|^4},
\]
the escort readout with $f(t)=t^2$ in the notation of Section~\ref{sec:corollaries}.
Its square-root readout has coordinates
$R_{M,i}=|\braket{e_i}{\psi}|^2/\bigl(\sum_j|\braket{e_j}{\psi}|^4\bigr)^{1/2}$,
which are smooth on $\CP^{d-1}$; hence (H1) holds, and calibration gives (H3).
It violates (H2): writing $T(u)_i=u_i^2/\sum_j u_j^2$ on the simplex, at the barycenter $u^\ast=(1/d,\ldots,1/d)$ one computes $\dd T|_{u^\ast}\,v=2v$ on simplex tangent vectors, so $T$ fixes $u^\ast$ and quadruples the Fisher quadratic form there.
Choose $v\neq0$ and a full-support constant-phase smooth curve through the barycenter with tangent $v$.
It saturates $F_{\mathrm{cl}}=F_Q$ for Born by Lemma~\ref{lem:phase-variance}, while the quartic readout has $F_{\mathrm{cl}}=4F_Q>F_Q$ at $s=0$.
The two Fisher speeds vary continuously near the barycenter, so the strict inequality persists on an open parameter interval and violates the almost-everywhere condition (H2).

\smallskip
\noindent\emph{Staircase readout (fails (H1) only; $d=2$).}
Fix $d=2$, write $\theta=d_{\mathrm{FS}}([\psi],[e_1])\in[0,\pi/2]$, and fix $0<\theta_1<\theta_2<\pi/2$.
Let $C:[0,1]\to[0,1]$ be the Cantor function \cite{Dovgoshey2006}, and let $c:[0,\pi/2]\to[0,\pi/2]$ equal the identity outside $(\theta_1,\theta_2)$ and
\[
c(\theta)=\theta_1+(\theta_2-\theta_1)\,C\!\left(\frac{\theta-\theta_1}{\theta_2-\theta_1}\right)
\]
inside, so that $c$ is continuous, non-decreasing, and singular on the band ($c'=0$ almost everywhere there). Set
\[
P_M([\psi])=\bigl(\cos^2 c(\theta),\ \sin^2 c(\theta)\bigr).
\]
Because the outer pieces are Lipschitz and match the rescaled Cantor piece at $\theta_1$ and $\theta_2$, after enlarging the H\"older constant the patched map $c$ is globally $\alpha$-H\"older on $[0,\pi/2]$, where $\alpha=\log 2/\log 3$ \cite{Dovgoshey2006}.
This readout is continuous on $\CP^{1}$, calibrated ($c(0)=0$ and $c(\pi/2)=\pi/2$), and not Born.
The upper-metric-speed definition \eqref{eq:upper-speed} does not presuppose absolute continuity, and this readout satisfies (H2) along every smooth curve.
Indeed, $R_M\circ\gamma=(\cos(c\circ\theta),\sin(c\circ\theta))$, so its round upper metric speed equals $\operatorname{md}^{+}(c\circ\theta)$.
For parameters with $\theta(s)\notin[\theta_1,\theta_2]$, $c$ is locally the identity and this upper speed is at most $|\theta'(s)|\leq\norm{\gamma'(s)}_{\mathrm{FS}}$ almost everywhere, since the distance function $\theta$ is $1$-Lipschitz along curves.
For parameters with $\theta(s)\in[\theta_1,\theta_2]$, the curve is away from $[e_1]$ and $[e_2]$, so $\theta(s)$ is smooth: where $\theta'\neq 0$, the inverse function theorem makes $\theta$ locally bi-Lipschitz, so the preimage of the null set on which $c'$ fails to vanish is null, and the derivative of $c\circ\theta$ exists and equals $c'(\theta)\,\theta'=0$ almost everywhere; where $\theta'=0$, the Taylor estimate $|\Delta\theta|=O(h^2)$ combined with the global $\alpha$-H\"older modulus of $c$ gives $|\Delta(c\circ\theta)|=O(h^{2\alpha})$ with $2\alpha>1$, hence upper metric speed zero.
Together with $F_Q(\gamma;s)=4\norm{\gamma'(s)}_{\mathrm{FS}}^2$, these estimates prove (H2).
Yet $R_M\circ\gamma$ is not absolutely continuous along the $[e_1]$--$[e_2]$ geodesic: the staircase transports the readout across the band while registering zero Fisher velocity almost everywhere, so (H1) fails.
Consistently with Corollary~\ref{cor:contrapositive} below, this readout does violate the global bound \eqref{eq:global-lip} (over the $n$th-stage Cantor intervals it expands round distance against Fubini--Study distance at rate $(3/2)^n$) while keeping the almost-everywhere differential bound; that combination is exactly what failing (H1) permits.
Absolute continuity is what upgrades almost-everywhere velocity bounds to the length bound in the proof of Theorem~\ref{thm:pvm-readout}; it cannot be weakened to continuity.

\smallskip
These examples show that none of the three hypotheses follows from the other two: (H2) is metric admissibility (the quartic escort fails it while smooth and calibrated), (H3) is labeling (uniform and permuted Born fail it while metrically admissible), and (H1) is the regularity that gives the metric condition global grip (the staircase readout moves without registering almost-everywhere Fisher velocity).
Born satisfies all three (Lemma~\ref{lem:phase-variance}), and Theorem~\ref{thm:pvm-readout} shows nothing else does.

\medskip
\noindent\textbf{Defending (H3): the equivalence-class route is the minimal route.}
Two routes support (H3) operationally. The minimal one is definitional:
a PVM apparatus with outcomes labeled by the rank-1 projectors $\{\ket{e_i}\bra{e_i}\}$ is, by definition, the equivalence class of operational devices whose statistics on the basis preparations $\ket{e_i}$ are the indicator distributions $\delta_i$ for every $i$.
Under this definition (H3) is part of what it means to call the apparatus $M$, restricted to basis-state preparations only.
A second route appeals to repeatability and the von Neumann projection postulate; it reaches the same calibration but under a stronger post-measurement assumption.
I treat the equivalence-class route as the minimal route and the repeatability route as a secondary sanity check.
Both routes are analyzed in standard measurement-theory references; see von Neumann \cite{vonNeumann1932}, Busch--Lahti--Mittelstaedt \cite{BuschLahtiMittelstaedt1996}, and Ozawa \cite{Ozawa1984}.

The conjunction (H1)--(H3) is then: the local speed bound has global force, the readout creates no distinguishability beyond what the projective geometry supplies, and the labels mean what they say.
Theorem~\ref{thm:pvm-readout} below shows that nothing else is admissible.

\section{Vertex Rigidity Lemma}

The geometric closure step is a sphere statement.

\begin{lemma}[Vertex Rigidity]\label{lem:vertex-rigidity}
Let
\[
\Psi:S^{d-1}_+\to S^{d-1}_+
\]
be round-metric $1$-Lipschitz and satisfy
\[
\Psi(e_i)=e_i,\qquad i=1,\ldots,d,
\]
for the coordinate vertices $e_i$. Then $\Psi=\id$.
\end{lemma}

\begin{proof}
Fix $x\in S^{d-1}_+$. Since $\Psi$ is $1$-Lipschitz and fixes the vertices,
\begin{equation}\label{eq:lip-to-vertex}
 d_{\mathrm{round}}(\Psi(x),e_i)\leq d_{\mathrm{round}}(x,e_i).
\end{equation}
On the unit sphere,
\begin{equation}\label{eq:cos-identity}
 \cos d_{\mathrm{round}}(y,e_i)=y_i
\end{equation}
for any $y\in S^{d-1}_+$. Cosine is strictly decreasing on $[0,\pi]$, so \eqref{eq:lip-to-vertex} and \eqref{eq:cos-identity} give
\begin{equation}\label{eq:dominance}
 \Psi(x)_i\geq x_i,\qquad i=1,\ldots,d.
\end{equation}
Both $x$ and $\Psi(x)$ lie on the unit sphere, hence
\begin{equation}\label{eq:unit-norm}
 \sum_i \Psi(x)_i^2=1=\sum_i x_i^2.
\end{equation}
Componentwise dominance plus equality of the sums of squares forces $\Psi(x)_i=x_i$ for every $i$.
Subtract \eqref{eq:unit-norm}:
\[
0=\sum_i\bigl(\Psi(x)_i^2-x_i^2\bigr)
 =\sum_i\bigl(\Psi(x)_i-x_i\bigr)\bigl(\Psi(x)_i+x_i\bigr).
\]
Each summand is non-negative: the first factor is non-negative by \eqref{eq:dominance}, and the second is non-negative because both $\Psi(x)_i$ and $x_i$ lie in the closed positive orthant.
A sum of non-negative terms equals zero only when every term is zero.
If $\Psi(x)_i+x_i>0$, then $\Psi(x)_i=x_i$; if $\Psi(x)_i+x_i=0$, then both terms are zero because neither is negative.
Therefore $\Psi(x)=x$ for all $x$, so $\Psi=\id$.
\end{proof}

The closing argument is elementary spherical geometry, and the lemma should be read as a packaging of standard facts about the round orthant, stated in the form the readout theorem needs.
The adjacent characterization literature answers a different question: Cencov \cite{Cencov1981} characterizes the classical Fisher metric on probability simplices, up to scale, by invariance under congruent Markov embeddings, and Campbell \cite{Campbell1986} extends this characterization framework to positive cones; Petz \cite{Petz1996} instead characterizes a non-unique family of monotone metrics on positive matrix spaces via operator-monotone functions. These are metric characterization or classification results, not rigidity theorems for non-expanding self-maps. For the surrounding geometry see \cite{AmariNagaoka2000,BengtssonZyczkowski2006}.

\section{Simplex Fisher Contraction Rigidity}

The simplex theorem isolates the pure information-geometric engine.
It is the square-root-chart formulation of the same spherical rigidity: the main readout theorem below runs the vertex argument directly on the sphere, and the simplex form is what the escort corollaries of Section~\ref{sec:corollaries} consume.

\begin{theorem}[Simplex Fisher contraction rigidity]\label{thm:simplex-rigidity}
Let
\[
T:\Delta^{d-1}\to\Delta^{d-1}
\]
be continuous on the closed simplex, interior-preserving,
$T(\Delta^{\circ}_{d-1})\subseteq\Delta^{\circ}_{d-1}$,
$C^1$ on $\Delta^{\circ}_{d-1}$, and Fisher non-expanding in the sense that, for every $u\in\Delta^{\circ}_{d-1}$ and every $v\in\R^d$ with $\sum_i v_i=0$,
\begin{equation}\label{eq:fisher-nonexp}
 g^F_{T(u)}\bigl(\dd T|_u\cdot v,\dd T|_u\cdot v\bigr)\leq g^F_u(v,v),
\end{equation}
and vertex-fixing,
\begin{equation}\label{eq:vertex-fix}
 T(e_i)=e_i,\qquad i=1,\ldots,d.
\end{equation}
Then $T=\id$ on $\Delta^{d-1}$.
\end{theorem}

Interior preservation is what makes the left side of \eqref{eq:fisher-nonexp} well-defined: the Fisher form \eqref{eq:fisher-metric} lives on the open simplex, and the differentiability of the conjugated map below is taken at interior image points.

\begin{proof}
Define the square-root chart
\[
\Phi(u)=(\sqrt{u_1},\ldots,\sqrt{u_d})
\]
from $\Delta^{\circ}_{d-1}$ to the open positive spherical orthant
\[
S^{d-1,\circ}_+=\left\{x\in\R^d:x_i>0,\ \sum_i x_i^2=1\right\}.
\]
For tangent vectors $v$,
\begin{equation}\label{eq:chart-isometry}
4\inner{\dd\Phi|_u\cdot v}{\dd\Phi|_u\cdot v}_{\mathrm{round}}=g^F_u(v,v).
\end{equation}
So the conjugate map
\[
\Psi=\Phi\circ T\circ\Phi^{-1}
\]
is infinitesimally $1$-Lipschitz on $S^{d-1,\circ}_+$ for the round metric.
For interior points $x,y\in S^{d-1,\circ}_+$, take the normalized chord
\[
\gamma(t)=\frac{(1-t)x+ty}{\norm{(1-t)x+ty}},\qquad t\in[0,1].
\]
Every component of $(1-t)x+ty$ is strictly positive, so $\gamma$ stays in $S^{d-1,\circ}_+$ and the denominator never vanishes.
Moreover $x\cdot y=\sum_i x_i y_i>0$, hence $x$ and $y$ lie in a common open hemisphere.
Therefore $\gamma$ parameterizes the unique minimizing great-circle arc, with $L(\gamma)=d_{\mathrm{round}}(x,y)$.
Integrating the differential inequality along $\gamma$ gives
\begin{equation}\label{eq:round-lip}
 d_{\mathrm{round}}(\Psi(x),\Psi(y))\leq d_{\mathrm{round}}(x,y)
\end{equation}
for interior points.
Continuity of $T$ extends $\Psi$ to the closed orthant $S^{d-1}_+$.
For boundary points $x',y'\in S^{d-1}_+$, choose interior sequences $x_n\to x'$ and $y_n\to y'$ and pass to the limit in \eqref{eq:round-lip} using continuity of $d_{\mathrm{round}}$ and $\Psi$.
Thus \eqref{eq:round-lip} holds on the closed orthant as well. By \eqref{eq:vertex-fix}, $\Psi(e_i)=e_i$.
The Vertex Rigidity Lemma then gives $\Psi=\id$, hence $T=\id$ after conjugating back through $\Phi$.
\end{proof}

The theorem is dimensionwise and geometric. It says nothing by itself about quantum measurements until a readout map is tied to a simplex map or to a square-root map anchored at calibrated vertices.

\section{PVM Readout Rigidity}

The proof has two moves: integrate the almost-everywhere speed bound of (H2) along minimizing geodesics into a global Lipschitz bound, then let the calibrated vertices force equality coordinate by coordinate.

\begin{theorem}[PVM readout rigidity]\label{thm:pvm-readout}
Fix $d\geq 2$ and the rank-1 PVM $M=\{\ket{e_i}\bra{e_i}\}$.
Let $P_M:\CP^{d-1}\to\Delta^{d-1}$ be a readout map such that:
\begin{enumerate}
\item $R_M=\sqrt{P_M}$ is continuous on $\CP^{d-1}$ and absolutely continuous along affinely parameterized Fubini--Study geodesic segments;
\item for every smooth pure-state curve $\gamma:I\to\CP^{d-1}$,
\[
F_{\mathrm{cl}}(P_M\circ\gamma;s)\leq F_Q(\gamma;s)
\qquad\text{for almost every }s\in I;
\]
\item $P_M([e_i])=\delta_i$ for $i=1,\ldots,d$.
\end{enumerate}
Then
\[
P_M([\psi])_i=|\braket{e_i}{\psi}|^2
\]
for every pure state $[\psi]\in\CP^{d-1}$ and every $i$.
\end{theorem}

\begin{proof}
Choose any minimizing Fubini--Study geodesic $\gamma:[a,b]\to\CP^{d-1}$ from $[\psi]$ to $[\phi]$.
Such geodesics exist because $\CP^{d-1}$ is compact and complete; at the cut locus they need not be unique, but each has length $d_{\mathrm{FS}}([\psi],[\phi])$.
By (H1), $r=R_M\circ\gamma$ is absolutely continuous, so its ordinary metric derivative exists almost everywhere and agrees there with the upper metric speed in \eqref{eq:upper-speed}.
Equations \eqref{eq:fq-fs} and \eqref{eq:upper-speed}, together with (H2), therefore give
\begin{equation}\label{eq:speed-bound}
\operatorname{md}r(s)
\leq
\norm{\gamma'(s)}_{\mathrm{FS}}
\qquad\text{for almost every }s\in[a,b].
\end{equation}
The length formula \eqref{eq:length-metric-derivative} and \eqref{eq:speed-bound} give
\[
\begin{aligned}
 d_{\mathrm{round}}(R_M([\psi]),R_M([\phi]))
 &\leq L_{\mathrm{round}}(r)\\
 &= \int_a^b \operatorname{md}r(s)\,\dd s\\
 &\leq \int_a^b \norm{\gamma'(s)}_{\mathrm{FS}}\,\dd s\\
 &=L_{\mathrm{FS}}(\gamma)=d_{\mathrm{FS}}([\psi],[\phi]).
\end{aligned}
\]
Therefore
\begin{equation}\label{eq:global-lip}
 d_{\mathrm{round}}(R_M([\psi]),R_M([\phi]))\leq d_{\mathrm{FS}}([\psi],[\phi])
\end{equation}
for all pure states $[\psi],[\phi]$.
Set $[\phi]=[e_i]$. By (H3),
\begin{equation}\label{eq:calib-vertex}
 R_M([e_i])=e_i.
\end{equation}
The Fubini--Study distance to a basis ray is
\begin{equation}\label{eq:fs-vertex-distance}
 d_{\mathrm{FS}}([\psi],[e_i])=\arccos\bigl(|\braket{e_i}{\psi}|\bigr).
\end{equation}
Combining \eqref{eq:global-lip}, \eqref{eq:calib-vertex}, and \eqref{eq:fs-vertex-distance},
\begin{equation}\label{eq:vertex-bound}
 d_{\mathrm{round}}(R_M([\psi]),e_i)\leq \arccos\bigl(|\braket{e_i}{\psi}|\bigr).
\end{equation}
Applying cosine gives
\begin{equation}\label{eq:cos-dominance}
 R_M([\psi])_i\geq |\braket{e_i}{\psi}|,\qquad i=1,\ldots,d.
\end{equation}
Both vectors $R_M([\psi])$ and $\bigl(|\braket{e_i}{\psi}|\bigr)_i$ lie in $S^{d-1}_+$, and both have squared coordinates summing to $1$.
Therefore the componentwise inequalities in \eqref{eq:cos-dominance} are equalities:
\begin{equation}\label{eq:amp-equal}
 R_M([\psi])_i=|\braket{e_i}{\psi}|.
\end{equation}
Squaring \eqref{eq:amp-equal} yields
\begin{equation}\label{eq:born}
 P_M([\psi])_i=|\braket{e_i}{\psi}|^2.
\end{equation}
So the readout is the Born readout for the fixed PVM $M$.
\end{proof}

\begin{corollary}[Calibrated non-Born readouts violate metric non-expansion]\label{cor:contrapositive}
Fix $d\geq 2$ and the rank-1 PVM $M=\{\ket{e_i}\bra{e_i}\}$, and let $P_M:\CP^{d-1}\to\Delta^{d-1}$ be a readout satisfying calibration (H3) but disagreeing with Born. Then there exist a pure state $[\psi]$ and a calibrated basis ray $[e_k]$ such that
\[
 d_{\mathrm{round}}(R_M([\psi]),R_M([e_k])) > d_{\mathrm{FS}}([\psi],[e_k]).
\]
In particular, if such a readout also satisfies regularity (H1), then along a minimizing geodesic from $[\psi]$ to $[e_k]$ the inequality $F_{\mathrm{cl}}>F_Q$ holds on a set of positive measure, so (H2) fails.
\end{corollary}

\begin{proof}
Choose $[\psi]$ where the readout differs from Born and set $r_i=R_M([\psi])_i$ and $a_i=|\braket{e_i}{\psi}|$.
The distinct vectors $r,a\in S^{d-1}_+$ have equal unit norm, so $r_k<a_k$ for some $k$.
Calibration gives $R_M([e_k])=e_k$, and hence
\[
 d_{\mathrm{round}}(R_M([\psi]),R_M([e_k]))
 =\arccos r_k
 >\arccos a_k
 =d_{\mathrm{FS}}([\psi],[e_k]).
\]
If (H1) and the almost-everywhere inequality (H2) held along a minimizing geodesic joining these endpoints, the integration argument \eqref{eq:speed-bound}$\to$\eqref{eq:global-lip} would give the opposite distance inequality.
Thus $F_{\mathrm{cl}}>F_Q$ on a positive-measure subset of that geodesic, which is precisely a failure of (H2).
\end{proof}

Stated contrapositively, then, Theorem~\ref{thm:pvm-readout} says that once the apparatus is calibrated, Born is the only readout that avoids a metric-non-expansion violation: any calibrated pure-state readout departing from Born must somewhere manufacture statistical distinguishability that the projective Fubini--Study geometry does not supply.

This theorem is the main claim of the paper.
It is not a Gleason-type theorem \cite{Gleason1957}: the input is a fixed PVM with calibrated eigenstates, square-root regularity, and a universal Fisher non-expansion condition on its readout, rather than a basis-independent measure on projectors.

\begin{remark}[Hypotheses actually used]\label{rem:minimal}
The displayed proof derives the global bound \eqref{eq:global-lip}, but the Born conclusion needs only its specialization to calibrated vertices.
Applying the same integration directly along minimizing Fubini--Study geodesics with an endpoint in $\{[e_1],\ldots,[e_d]\}$ shows that Theorem~\ref{thm:pvm-readout} remains valid with (H1) and (H2) restricted to that family of curves.
The integrated inequality the proof extracts from them is the vertexwise distance bound
\begin{equation}\label{eq:distance-form}
 d_{\mathrm{FR}}\bigl(P_M([\psi]),\delta_i\bigr)\leq 2\,d_{\mathrm{FS}}([\psi],[e_i]),
 \qquad i=1,\ldots,d,
\end{equation}
where $d_{\mathrm{FR}}$ is the completed Fisher--Rao distance on the closed simplex (the Fisher metric itself lives on the open simplex, and its metric completion is the round-orthant geometry of the square-root chart).
By \eqref{eq:chart-isometry} the Fisher metric is four times the round metric in the square-root chart, so $d_{\mathrm{FR}}$ is twice the round distance (the positive orthant is geodesically convex, so intrinsic and ambient round distances agree), and \eqref{eq:distance-form} is exactly \eqref{eq:vertex-bound}: Wootters' statistical angle $\arccos\sqrt{P_M([\psi])_i}=\tfrac12\,d_{\mathrm{FR}}\bigl(P_M([\psi]),\delta_i\bigr)$ is bounded by the quantum angle $\arccos|\braket{e_i}{\psi}|$ \cite{Wootters1981}.
The Born readout attains equality in \eqref{eq:distance-form} for every $[\psi]$ and every $i$.
Inequality \eqref{eq:distance-form} together with normalization already forces Born, by the steps \eqref{eq:vertex-bound}--\eqref{eq:born}; stated against the named vertices $\delta_i$, it absorbs the calibration anchor.
Operationally, \eqref{eq:distance-form} says the readout may not report $[\psi]$ as more distinguishable from any calibration preparation, in statistical angle, than the projective geometry already makes it.
The role of the curve-level primitives (H1)--(H3) is to derive \eqref{eq:distance-form} rather than postulate it.
\end{remark}

\section{Calibration as Apparatus Labeling}\label{sec:calibration}

The status of (H3) was set out in Section~\ref{sec:hypotheses}: it is the boundary anchor on the $d$ basis preparations, supplied by the equivalence-class definition of the labeled PVM apparatus, and weaker than the full Born formula.
Two further remarks are in order.

First, (H3) is not a Fisher consequence. The uniform readout from Section~\ref{sec:hypotheses} satisfies (H2) trivially yet violates (H3); no Fisher-only hypothesis on $P_M$ can supply the calibration condition.
That counterexample makes (H3) physical input rather than redundant decoration of (H2).

Second, the labeling content of (H3) is semantic, not dynamical. The permuted-Born readout from Section~\ref{sec:hypotheses} satisfies (H2) (it equals Born up to relabeling, Lemma~\ref{lem:phase-variance}) yet violates (H3) for the labeled PVM $M$.
(H3) is what fixes the labels of the apparatus to the projectors they are declared to measure.
Once the labels mean what they say, and once the readout satisfies square-root regularity (H1) and universal Fisher non-expansion (H2), the only admissible map on the rest of $\CP^{d-1}$ is Born by Theorem~\ref{thm:pvm-readout}.

\section{Arbitrary-Rank Projective Measurements}\label{sec:rank-general}

The rank-1 restriction can be removed. Fix $d\geq 2$ and a projective measurement
\[
M=\{\Pi_1,\ldots,\Pi_m\},\qquad
\Pi_i^\dagger=\Pi_i=\Pi_i^2,\quad
\Pi_i\Pi_j=0\ (i\neq j),\quad
\sum_{i=1}^m\Pi_i=I,
\]
with $m\geq 2$ outcomes and ranks $k_i=\operatorname{rank}\Pi_i\geq 1$, $\sum_i k_i=d$.
A readout is now a map $P_M:\CP^{d-1}\to\Delta^{m-1}$, with square-root readout $R_M=\sqrt{P_M}:\CP^{d-1}\to S^{m-1}_+$, and the definitions of Section~\ref{sec:setup} (the upper metric speed \eqref{eq:upper-speed}, $F_{\mathrm{cl}}$, and $F_Q$) apply verbatim with the output index set $\{1,\ldots,m\}$.
Hypotheses (H1) and (H2) are unchanged.
Calibration is imposed on each eigenspace:

\medskip
\noindent\textbf{(H3$'$) Eigenspace calibration.}
For every $i$ and every unit vector $\chi\in\operatorname{range}\Pi_i$,
\[
P_M([\chi])=\delta_i.
\]

\medskip
For $k_i\equiv 1$ this is (H3).
As in Section~\ref{sec:hypotheses}, (H3$'$) is read as apparatus labeling: an outcome labeled by the projector $\Pi_i$ is, definitionally, reported with certainty on every preparation in $\operatorname{range}\Pi_i$.
Note that calibrating only an orthonormal basis of each eigenspace would not suffice for the argument below: the extremal anchor $[\Pi_i\psi/\norm{\Pi_i\psi}]$ ranges over the whole subspace as $[\psi]$ varies.

The calibrated set is the disjoint union of the projective subspaces $\mathbb{P}(\operatorname{range}\Pi_i)\cong\CP^{k_i-1}$, and the Fubini--Study distance from a state to the $i$th subspace has the closed form
\begin{equation}\label{eq:subspace-distance}
 d_{\mathrm{FS}}\bigl([\psi],\mathbb{P}(\operatorname{range}\Pi_i)\bigr)=\arccos\norm{\Pi_i\psi},
\end{equation}
attained, when $\Pi_i\psi\neq 0$, at the normalized projection $[\Pi_i\psi/\norm{\Pi_i\psi}]$, the L\"uders state for outcome $i$ \cite{BuschLahtiMittelstaedt1996}.
Indeed, for unit $\chi\in\operatorname{range}\Pi_i$ one has $\braket{\chi}{\psi}=\braket{\chi}{\Pi_i\psi}$, and Cauchy--Schwarz gives $|\braket{\chi}{\psi}|\leq\norm{\Pi_i\psi}$, with equality exactly at the normalized projection.

\begin{lemma}[Born membership, arbitrary rank]\label{lem:born-membership-general}
Let $B_M([\psi])_i=\braket{\psi}{\Pi_i\psi}$ for a normalized representative $\psi$.
\begin{enumerate}
\item $B_M$ satisfies (H1), (H2), and (H3$'$).
\item Let $[\psi(s)]$ be a smooth pure-state curve with horizontal representative, $\braket{\psi}{\dot\psi}=0$, along which $a_i(s)=\norm{\Pi_i\psi(s)}>0$ for all $i$. Then each $a_i$ is smooth and
\begin{equation}\label{eq:block-deficit}
 F_Q=F_{\mathrm{cl}}(B_M\circ\gamma)+4\sum_{i=1}^m\bigl(\norm{\Pi_i\dot\psi}^2-\dot a_i^{\,2}\bigr),
\end{equation}
with every summand non-negative; $F_{\mathrm{cl}}=F_Q$ at a parameter $s$ if and only if $\Pi_i\dot\psi(s)\in\R\,\Pi_i\psi(s)$ for every $i$.
At rank one, in the horizontal gauge $\sum_i p_i\dot\varphi_i=0$, identity \eqref{eq:block-deficit} is the phase-variance identity \eqref{eq:phase-variance}.
\end{enumerate}
\end{lemma}

\begin{proof}
(H3$'$) is immediate: $\Pi_j\chi=\delta_{ij}\chi$ for unit $\chi\in\operatorname{range}\Pi_i$.
For (H1), the reverse triangle inequality gives
$\bigl|\norm{\Pi_i\psi(s)}-\norm{\Pi_i\psi(t)}\bigr|\leq\norm{\Pi_i(\psi(s)-\psi(t))}\leq\norm{\psi(s)-\psi(t)}$
for any representative, so each $a_i$ is locally Lipschitz along smooth curves, hence absolutely continuous on compact parameter intervals; continuity of $B_M$ on $\CP^{d-1}$ is clear.
For (H2), choose a horizontal representative on any parameter interval, so that
$F_Q=4\norm{\dot\psi}^2=4\sum_i\norm{\Pi_i\dot\psi}^2$
by completeness and mutual orthogonality of the $\Pi_i$.
The output amplitude curve $(a_i)_i$ is absolutely continuous, so \eqref{eq:upper-speed} agrees almost everywhere with its ordinary round velocity.
At almost every parameter: where $a_i(s)>0$ and $a_i$ is differentiable,
\[
\dot a_i=\frac{\mathrm{Re}\,\braket{\Pi_i\psi}{\Pi_i\dot\psi}}{a_i},
\qquad
|\dot a_i|\leq\norm{\Pi_i\dot\psi}
\]
by Cauchy--Schwarz, while at almost every point of the set $\{a_i=0\}$ where $a_i$ is differentiable, $\dot a_i=0$.
Hence $F_{\mathrm{cl}}=4\sum_i\dot a_i^{\,2}\leq 4\sum_i\norm{\Pi_i\dot\psi}^2=F_Q$ almost everywhere.

For \eqref{eq:block-deficit}, positivity of every $a_i$ along the curve makes $a_i=\braket{\psi}{\Pi_i\psi}^{1/2}$ smooth, and subtracting the two displays above gives
$F_Q-F_{\mathrm{cl}}=4\sum_i\bigl(\norm{\Pi_i\dot\psi}^2-\dot a_i^{\,2}\bigr)$.
Writing $u_i=\Pi_i\psi/a_i$, each summand equals $\norm{\Pi_i\dot\psi}^2-\bigl(\mathrm{Re}\braket{u_i}{\Pi_i\dot\psi}\bigr)^2\geq 0$, and it vanishes precisely when $\Pi_i\dot\psi$ is a real multiple of $u_i$, i.e.\ $\Pi_i\dot\psi\in\R\,\Pi_i\psi$.
At rank one, $\Pi_i\dot\psi=\dot z_i\ket{e_i}$ with $z_i=\sqrt{p_i}\,e^{i\varphi_i}$, and
$\norm{\Pi_i\dot\psi}^2-\dot a_i^{\,2}=|\dot z_i|^2-(\dd|z_i|/\dd s)^2=p_i\dot\varphi_i^2$;
summing gives $4\sum_i p_i\dot\varphi_i^2=4\Var_p(\dot\varphi)$ in the horizontal gauge $\sum_i p_i\dot\varphi_i=0$.
\end{proof}

\begin{theorem}[Projective readout rigidity, arbitrary rank]\label{thm:rank-general}
Fix $d\geq 2$ and the projective measurement $M=\{\Pi_1,\ldots,\Pi_m\}$, $m\geq 2$.
Let $P_M:\CP^{d-1}\to\Delta^{m-1}$ satisfy (H1), (H2), and (H3$'$).
Then
\[
P_M([\psi])_i=\braket{\psi}{\Pi_i\psi}
\]
for every pure state $[\psi]\in\CP^{d-1}$ and every $i$.
\end{theorem}

\begin{proof}
The integration argument \eqref{eq:speed-bound}$\to$\eqref{eq:global-lip} in the proof of Theorem~\ref{thm:pvm-readout} uses only (H1), (H2), and the two length structures, so it holds verbatim with output sphere $S^{m-1}_+$:
\[
d_{\mathrm{round}}(R_M([\psi]),R_M([\phi]))\leq d_{\mathrm{FS}}([\psi],[\phi])
\qquad\text{for all pure states }[\psi],[\phi].
\]
Fix $[\psi]$ and $i$ with $\Pi_i\psi\neq 0$, and set $\chi_i=\Pi_i\psi/\norm{\Pi_i\psi}$.
Then $\braket{\chi_i}{\psi}=\norm{\Pi_i\psi}$, so $d_{\mathrm{FS}}([\psi],[\chi_i])=\arccos\norm{\Pi_i\psi}$ by \eqref{eq:subspace-distance}, and (H3$'$) gives $R_M([\chi_i])=e_i$.
The global bound then yields
\[
 d_{\mathrm{round}}(R_M([\psi]),e_i)\leq\arccos\norm{\Pi_i\psi},
\]
and the cosine identity \eqref{eq:cos-identity} gives
\[
 R_M([\psi])_i\geq\norm{\Pi_i\psi}.
\]
If $\Pi_i\psi=0$ the inequality $R_M([\psi])_i\geq 0=\norm{\Pi_i\psi}$ is automatic.
The vectors $R_M([\psi])$ and $\bigl(\norm{\Pi_i\psi}\bigr)_{i=1}^m$ both lie in $S^{m-1}_+$: the first because $P_M([\psi])\in\Delta^{m-1}$, the second because $\sum_i\norm{\Pi_i\psi}^2=\norm{\psi}^2=1$ by completeness and mutual orthogonality.
Componentwise dominance together with equal sums of squares forces equality componentwise, exactly as in \eqref{eq:cos-dominance}--\eqref{eq:amp-equal}: $R_M([\psi])_i=\norm{\Pi_i\psi}$ for every $i$.
Squaring gives the claim.
\end{proof}

\begin{remark}[Anchors used at arbitrary rank]\label{rem:anchors-general}
As in Remark~\ref{rem:minimal}, the proof uses (H1) and (H2) only along minimizing geodesics with one endpoint in the calibrated set $\bigsqcup_i\mathbb{P}(\operatorname{range}\Pi_i)$, and the integrated content is the eigenspace distance bound
\[
 d_{\mathrm{FR}}\bigl(P_M([\psi]),\delta_i\bigr)\leq 2\,d_{\mathrm{FS}}\bigl([\psi],\mathbb{P}(\operatorname{range}\Pi_i)\bigr),
 \qquad i=1,\ldots,m,
\]
which the Born readout saturates for every $[\psi]$ and every $i$: by \eqref{eq:subspace-distance}, both sides equal $2\arccos\norm{\Pi_i\psi}$ for $B_M$.
Wootters' angle identity thus persists at arbitrary rank, with the L\"uders state as the extremal anchor.
\end{remark}

\begin{corollary}[Born assignment across all contexts]\label{cor:all-contexts}
Suppose that to every projective measurement $M$ on $\C^d$ a readout $P_M:\CP^{d-1}\to\Delta^{m(M)-1}$ is assigned, and that each $P_M$ separately satisfies (H1), (H2), and (H3$'$) relative to its own $M$.
Then $P_M([\psi])_i=\braket{\psi}{\Pi_i\psi}$ for every $M$, every outcome $i$, and every pure state $[\psi]$.
In particular the assignment is noncontextual (a projector occurring in two different measurements receives the same probability in both) as a consequence of the theorem rather than a premise, and the conclusion holds for every $d\geq 2$.
\end{corollary}

\begin{proof}
Theorem~\ref{thm:rank-general} applies to each $M$ separately, and the Born value $\braket{\psi}{\Pi\psi}$ depends only on the projector $\Pi$.
\end{proof}

\begin{remark}[Contrast with Gleason]\label{rem:gleason-contrast}
Gleason's theorem \cite{Gleason1957} derives the Born measure from a single cross-context axiom, additivity of a frame function over every orthogonal resolution of the identity, and holds for $d\geq 3$ while failing at $d=2$, where frame functions are unconstrained beyond normalization on antipodal pairs.
Corollary~\ref{cor:all-contexts} assumes more per context and nothing across contexts: (H1), (H2), and (H3$'$) are metric-admissibility and calibration conditions on each apparatus separately, and no relation between the readouts of different measurements is postulated.
The two input sets are logically incomparable, and neither result subsumes the other.
The comparison isolates where the work happens: in Gleason's route the difficulty is the cross-context gluing, which is unavailable at $d=2$; here it is carried by the per-context non-expansion hypothesis, which is dimension-insensitive.
Cross-context agreement is not derived from less: it is inherited from per-context uniqueness.
A cross-measurement axiom does enter this paper, deliberately and with its content made explicit, only in the non-projective extension of Section~\ref{sec:povm-closure}.
\end{remark}

\section{POVM Boundary and Conditional Extensions}\label{sec:povm-closure}

Theorems~\ref{thm:pvm-readout} and~\ref{thm:rank-general} end at projective measurements, and the boundary is not an artifact of the proof technique.
This section proves that the metric hypotheses, even together with the one-sided certainty-of-occurrence calibration (H3$''$), are rigid exactly on projective measurements (Theorem~\ref{thm:povm-boundary}), and then records two conditional routes to finite POVMs in order to expose the additional identification principle each route spends.
The preliminary non-rigidity proposition uses only (H1)--(H2), while the boundary theorem also imposes (H3$''$).
Neither conditional extension follows from (H1)--(H3$''$), and neither is used anywhere in the projective results.

\subsection{Why Fisher non-expansion does not rigidify a fixed POVM}

A readout for a finite POVM $E=(E_a)_{a=1}^m$, $E_a\geq0$, $\sum_aE_a=I$, is a map $P_E:\CP^{d-1}\to\Delta^{m-1}$, with (H1) and (H2) read verbatim on the output index set $\{1,\ldots,m\}$.
Write $B_E([\psi])_a=\braket{\psi}{E_a\psi}$ for the Born readout of $E$ on normalized representatives.
Born membership persists at this generality:

\begin{lemma}[Born membership, arbitrary POVM]\label{lem:born-membership-povm}
For every finite POVM $E$ on $\C^d$, the Born readout $B_E$ satisfies (H1) and (H2).
If moreover $E_a>0$ for every $a$, its square-root readout is smooth with components bounded below by $\min_a\lambda_{\min}(E_a)^{1/2}>0$.
\end{lemma}

\begin{proof}
For a normalized representative, $\sqrt{p_a(s)}=\norm{E_a^{1/2}\psi(s)}$ with $p_a(s)=\braket{\psi(s)}{E_a\psi(s)}$, so the reverse triangle inequality gives
\[
\bigl|\sqrt{p_a(s)}-\sqrt{p_a(t)}\bigr|
\leq\norm{E_a^{1/2}\bigl(\psi(s)-\psi(t)\bigr)}
\leq\norm{\psi(s)-\psi(t)},
\]
using $\norm{E_a^{1/2}}\leq1$; each $\sqrt{p_a}$ is therefore locally Lipschitz along smooth curves, hence absolutely continuous on compact parameter intervals, and (H1) holds; continuity of $B_E$ on $\CP^{d-1}$ is clear.
For (H2), choose a horizontal representative, $\braket{\psi}{\dot\psi}=0$, so $F_Q=4\norm{\dot\psi}^2$.
The output amplitude curve $(\sqrt{p_a})_a$ is absolutely continuous, so \eqref{eq:upper-speed} agrees almost everywhere with its Euclidean velocity.
Where $p_a>0$, $\dot p_a=2\,\mathrm{Re}\braket{\dot\psi}{E_a\psi}$, and the Cauchy--Schwarz inequality for the semi-inner product $\langle\cdot,E_a\,\cdot\rangle$ gives $\dot p_a^{\,2}\leq 4\braket{\dot\psi}{E_a\dot\psi}\,p_a$, hence $(\dd\sqrt{p_a}/\dd s)^2\leq\braket{\dot\psi}{E_a\dot\psi}$; at almost every point of the set $\{p_a=0\}$ where $\sqrt{p_a}$ is differentiable, its derivative vanishes.
Summing over $a$ and using $\sum_aE_a=I$,
\[
F_{\mathrm{cl}}(B_E\circ\gamma)=4\sum_a\Bigl(\frac{\dd\sqrt{p_a}}{\dd s}\Bigr)^{\!2}
\leq 4\sum_a\braket{\dot\psi}{E_a\dot\psi}
=4\norm{\dot\psi}^2=F_Q
\qquad\text{almost everywhere.}
\]
At rank one this recovers part of Lemma~\ref{lem:phase-variance}, and for projectors the (H2) computation of Lemma~\ref{lem:born-membership-general}, since then $\braket{\dot\psi}{\Pi_a\dot\psi}=\norm{\Pi_a\dot\psi}^2$.
The positivity statement is immediate from $p_a\geq\lambda_{\min}(E_a)$.
\end{proof}

\begin{proposition}[Non-rigidity under the metric hypotheses alone]\label{prop:povm-nonrigidity}
Let $E=(E_a)_{a=1}^m$ be a finite POVM on $\C^d$ with $E\neq(I/m,\ldots,I/m)$.
For $0<\lambda<1$ define the barycentric contraction $T_\lambda(u)_a=(1-\lambda)u_a+\lambda/m$ and the readout
\[
P^{(\lambda)}_E=T_\lambda\circ B_E.
\]
Then $P^{(\lambda)}_E$ has a smooth square-root readout with components bounded below by $\sqrt{\lambda/m}$, satisfies (H2), with $F_{\mathrm{cl}}(P^{(\lambda)}_E\circ\gamma)\leq(1-\lambda)\,F_Q(\gamma)$ almost everywhere on every smooth pure-state curve, and differs from the Born readout $B_E$.
\end{proposition}

\begin{proof}
Along a smooth curve with horizontal representative, write $p_a=\braket{\psi}{E_a\psi}$ and $q_a=(1-\lambda)p_a+\lambda/m$.
Each $q_a$ is smooth with $q_a\geq\lambda/m$, so the square-root readout is smooth and the classical Fisher expression applies everywhere:
\[
\begin{aligned}
F_{\mathrm{cl}}\bigl(P^{(\lambda)}_E\circ\gamma\bigr)
&=\sum_a\frac{(1-\lambda)^2\dot p_a^{\,2}}{q_a}\\
&\leq 4(1-\lambda)\sum_a
  \frac{(1-\lambda)p_a}{q_a}\,
  \braket{\dot\psi}{E_a\dot\psi}\\
&\leq 4(1-\lambda)\sum_a\braket{\dot\psi}{E_a\dot\psi}
 =(1-\lambda)\,F_Q,
\end{aligned}
\]
using the Cauchy--Schwarz bound $\dot p_a^{\,2}\leq 4\,p_a\braket{\dot\psi}{E_a\dot\psi}$ from the proof of Lemma~\ref{lem:born-membership-povm}, valid for every effect including those with kernels, together with $(1-\lambda)p_a\leq q_a$.
Finally, $T_\lambda\circ B_E=B_E$ would force $B_E$ to be constantly the barycenter, i.e.\ $\braket{\psi}{E_a\psi}=1/m$ for every $[\psi]$ and $a$, i.e.\ $E_a=I/m$ for every $a$, which is excluded.
\end{proof}

For the excluded uniform POVM $E_a=I/m$, any constant readout with interior value $q\neq(1/m,\ldots,1/m)$ is smooth with $F_{\mathrm{cl}}=0$ and differs from Born.
So no finite POVM has an (H1)--(H2)-rigid readout.

Three remarks make the failure precise.
First, the rival is a Born readout in disguise:
\[
P^{(\lambda)}_E([\psi])_a
=\braket{\psi}{\bigl((1-\lambda)E_a+\tfrac{\lambda}{m}I\bigr)\psi}.
\]
With $\Lambda_\lambda(b|a)=(1-\lambda)\delta_{ba}+\lambda/m$, the effects in this display are those of the nontrivial stochastic smearing $\Lambda_{\lambda *}E$.
The fuzzy/post-processing order, post-processing cleanliness, and minimal sufficiency of POVMs study this kind of classical information loss and redundancy \cite{Kuramochi2015,HaapasaloPellonpaa2017}.
Those results presuppose the standard Born statistical experiment and are cited here only to locate the ambiguity; they are not inputs to the proof.
Projective calibration excludes the ambiguity because the simplex vertices reached on eigenstate preparations expose any admixed classical noise, while a generic unsharp POVM has no comparable boundary witness.
Post-processing cleanliness or minimal sufficiency of the operator tuple $E$ does not by itself restore the present readout rigidity: the hypotheses have not yet identified an arbitrary candidate map $P_E$ with the standard statistical experiment of $E$ outside the certainty data.
Second, informational completeness does not restore rigidity: an informationally complete POVM spans the space of Hermitian operators, hence is not the uniform POVM, and the proposition applies to it directly.
Third, $T_\lambda$ preserves the componentwise order of the output, so weakened anchor conditions of argmax type do not exclude the family.

The boundary is a characterization, built on the block structure of certainty subspaces.
The existence of an exact certainty anchor for outcome $a$ is the finite-dimensional attainment form of $\norm{E_a}=1$. The norm-$1$ property for general observables is formulated through outcome probabilities that can be made arbitrarily close to one; in finite dimension the operator norm is attained, so $\norm{E_a}=1$ yields a nonzero subspace $S_a=\ker(I-E_a)$. Hypothesis (H3$''$) then calibrates the candidate readout on every unit vector in that subspace \cite{HeinonenEtAl2003}.

\begin{lemma}[Certainty block decomposition]\label{lem:certainty-blocks}
Let $E=(E_a)_{a=1}^m$ be a finite POVM on $\C^d$ and $S_a=\ker(I-E_a)$ its certainty subspaces.
Then the $S_a$ are mutually orthogonal, and with $S=\bigoplus_aS_a$ and $K=S^{\perp}$ every effect block-diagonalizes as
\[
E_a=P_{S_a}\oplus A_a,
\]
where $P_{S_a}$ is the orthogonal projector onto $S_a$ and $(A_a)_a$ is a POVM on $K$ with $\norm{A_a}<1$ for every $a$.
Moreover $K=0$ if and only if $E$ is projective; in particular the certainty subspaces of a non-projective POVM do not span $\C^d$.
\end{lemma}

\begin{proof}
For unit $\chi\in S_a$, $\sum_b\braket{\chi}{E_b\chi}=1$ and $\braket{\chi}{E_a\chi}=1$ force $\braket{\chi}{E_b\chi}=0$, hence $E_b\chi=0$, for every $b\neq a$.
For unit $\chi'\in S_b$ with $b\neq a$, $\braket{\chi'}{\chi}=\braket{E_b\chi'}{\chi}=\braket{\chi'}{E_b\chi}=0$, so the $S_a$ are mutually orthogonal.
Each $E_a$ maps $S$ into $S$, acting as the identity on $S_a$ and as zero on every other $S_b$; being self-adjoint, it therefore maps $K$ into $K$, and the restrictions $A_a=E_a|_K$ are positive with $\sum_aA_a=I_K$.
If $A_a\xi=\xi$ for a unit $\xi\in K$ then $\xi\in S_a\cap K=0$, so $1$ is not an eigenvalue of $A_a$ and $\norm{A_a}<1$ in finite dimension.
If $K=0$ then $E_a=P_{S_a}$, a projective measurement; conversely a projective $E$ has $S_a=\operatorname{range}E_a$, which span $\C^d$.
\end{proof}

\begin{theorem}[Rigidity boundary]\label{thm:povm-boundary}
Fix a finite POVM $E=(E_a)_{a=1}^m$ on $\C^d$, $m\geq2$, and impose on candidate readouts $P_E:\CP^{d-1}\to\Delta^{m-1}$ the hypotheses (H1), (H2), and

\smallskip
\noindent\textbf{(H3$''$) Certainty-of-occurrence calibration.} $P_E([\chi])=\delta_a$ for every $a$ and every unit vector $\chi\in S_a=\ker(I-E_a)$.

\smallskip
\noindent This condition is intentionally one-sided: it calibrates preparations on which one named outcome occurs with certainty. It does not separately impose zero-probability calibration on $\ker E_a$ or certainty calibration for coarse-grained events $\sum_{a\in G}E_a$; those stronger anchor sets are outside Theorem~\ref{thm:povm-boundary}.

\smallskip
\noindent These hypotheses force $P_E=B_E$ if and only if $E$ is projective.
\end{theorem}

\begin{proof}
If $E$ is projective, repeat the global Lipschitz and calibrated-anchor argument from Theorem~\ref{thm:rank-general} for every nonzero effect $E_a$.
It gives $R_E([\psi])_a\geq\norm{E_a\psi}$ for each such $a$.
The squared right-hand sides sum to one, while all squared components of $R_E([\psi])$ also sum to one, so equality holds for every nonzero effect and every component corresponding to a zero effect vanishes.
Thus $P_E=B_E$, including when zero-probability dummy outcomes are listed.
If $E$ is not projective, $K\neq0$ by Lemma~\ref{lem:certainty-blocks}.
Choose $q$ in the open simplex $\Delta^{m-1,\circ}$ with $q_aI_K\neq A_a$ for some $a$ (possible, since the condition $A_a=q_aI_K$ for all $a$ pins at most one point $q$) and set
\[
\widetilde E_a=P_{S_a}\oplus\bigl((1-\lambda)A_a+\lambda q_a I_K\bigr),
\qquad 0<\lambda<1.
\]
Then $\widetilde E$ is a POVM: each term is positive and $\sum_a\widetilde E_a=P_S\oplus I_K=I$.
Its Born readout $B_{\widetilde E}$ satisfies (H1) and (H2) by Lemma~\ref{lem:born-membership-povm}, satisfies (H3$''$) for $E$ because $\braket{\chi}{\widetilde E_b\chi}=\delta_{ab}$ for unit $\chi\in S_a$, and differs from $B_E$ because $\widetilde E_a-E_a=0_S\oplus\lambda(q_aI_K-A_a)$ is a nonzero self-adjoint operator for some $a$.
So the hypotheses do not force $B_E$.
\end{proof}

\begin{remark}[What the boundary means operationally]\label{rem:povm-semantics}
Theorem~\ref{thm:povm-boundary} quantifies over maps carrying the label $E$, but (H3$''$) ties that label to the operator tuple only through the certainty subspaces.
For a PVM those subspaces span the Hilbert space and identify every effect.
For a non-projective POVM, the residual effects $A_a$ on $K$ are invisible to the occurrence-certainty anchors required by (H3$''$), and the hypotheses cannot distinguish $E$ from the modified tuple $\widetilde E$ used in the proof.
Thus the obstruction is semantic as well as geometric: before the Born rule is available, one must say what operational data make an unsharp apparatus the POVM $E$ rather than another tuple with the same occurrence-certainty behavior.
Defining the label by its full Born statistics would build the desired conclusion into the apparatus identification; using the one-sided behavior in (H3$''$) alone leaves the unsharp sector underdetermined.
Any genuine POVM extension must supply an independent identification principle for that sector.
\end{remark}

Projective measurements are therefore exactly the finite POVMs whose readouts are rigid under the metric hypotheses together with the certainty-of-occurrence calibration stated in (H3$''$).
The next two subsections record explicit additional principles that recover the usual finite-POVM assignment; they do not sharpen the metric theorem.

\subsection{Classical post-processing naturality}

For a column-stochastic matrix $\Lambda$ ($\Lambda(b|a)\geq0$, $\sum_b\Lambda(b|a)=1$), define the post-processed POVM $(\Lambda_*E)_b=\sum_a\Lambda(b|a)\,E_a$.
Operationally, $\Lambda_*E$ can be implemented by measuring $E$ and passing the recorded outcome through the classical channel $\Lambda$.

\medskip
\noindent\textbf{(N) Classical post-processing naturality.}
Readouts are assigned to every finite POVM on $\C^d$ as a function of the POVM, and for every finite POVM $E$ and every column-stochastic $\Lambda$,
\[
P_{\Lambda_*E}([\psi])=\Lambda\,P_E([\psi])
\qquad\text{for every pure state }[\psi].
\]

\medskip
For the literal implementation ``measure $E$, then apply $\Lambda$,'' the transformation of observed frequencies is ordinary probability calculus.
The substantive content of (N) is the further identification of that implementation with the readout assigned solely from the effect tuple $\Lambda_*E$, including implementations in which the same unsharp effects arise by a different physical route.
This is an implementation-functionality, or effect-noncontextuality, assumption at the level of probability assignments.
It parallels the operational-equivalence language used for unsharp measurements in generalized contextuality \cite{Spekkens2005}, although no ontological model or noncontextual hidden-variable claim is made here.
It is not a consequence of (H1)--(H3$''$).

\begin{theorem}[Conditional finite-POVM closure under post-processing naturality]\label{thm:povm-closure}
Suppose a pure-state readout $P_E:\CP^{d-1}\to\Delta^{m(E)-1}$ is assigned to every finite POVM $E$ on $\C^d$, the readout of every rank-one PVM satisfies (H1)--(H3), and (N) holds.
Then
\[
P_E([\psi])_a=\braket{\psi}{E_a\psi}
\]
for every finite POVM $E$, every outcome $a$, and every pure state $[\psi]$.
\end{theorem}

\begin{proof}
Fix $E$ and an outcome $a$.
The one-versus-rest merge $\Lambda^{(a)}$, sending $a\mapsto1$ and every other outcome to $2$, is deterministic and column-stochastic, with $\Lambda^{(a)}_*E=\{E_a,\,I-E_a\}$, so (N) gives
\[
P_E([\psi])_a=P_{\{E_a,I-E_a\}}([\psi])_1.
\]
Diagonalize $E_a=\sum_{j=1}^d\mu_j\ket{f_j}\bra{f_j}$ with $0\leq\mu_j\leq1$, and let $M=\{\ket{f_j}\bra{f_j}\}_{j=1}^d$ be an eigenbasis PVM\@.
The stochastic matrix $\Lambda(1|j)=\mu_j$, $\Lambda(2|j)=1-\mu_j$ satisfies $\Lambda_*M=\{E_a,I-E_a\}$, so (N) and Theorem~\ref{thm:pvm-readout} applied to $M$ give
\[
P_{\{E_a,I-E_a\}}([\psi])_1=\sum_j\mu_j\,|\braket{f_j}{\psi}|^2=\braket{\psi}{E_a\psi}.
\qedhere
\]
\end{proof}

\begin{remark}[Minimal instances of (N), and what the proof avoids]\label{rem:minimal-naturality}
Only two instances of (N) are used directly: deterministic one-versus-rest merges of an arbitrary POVM, and binary stochastic smearings of rank-one PVMs.
Full stochastic naturality is a concise sufficient axiom; these two bridge relations are the subset that carries this proof. Globally, deterministic naturality on all finite POVMs already implies the stochastic instances by Proposition~\ref{prop:naturality-additivity}; locally, the spectral-smearing identification remains the substantive bridge.
The POVM $E$ is never simulated by projective measurements: only its binary margins $\{E_a,I-E_a\}$ are factored through individual PVMs, and the margins determine the probability vector.
The proof is therefore untouched by the fact that not every POVM is projective-simulable on the original system \cite{Oszmaniec2017}.
Finally, (N) excludes the family of Proposition~\ref{prop:povm-nonrigidity}: mixing outputs toward the barycenter does not commute with outcome merges of unequal block sizes.
\end{remark}

\subsection{Exact relation to effect additivity and prior art}\label{subsec:effect-noncontextuality}

\begin{proposition}[Naturality versus effect additivity]\label{prop:naturality-additivity}
For an assignment $E\mapsto P_E$ on all finite POVMs, (N) is equivalent to the following statement: for every pure state $[\psi]$ there is a map
\[
v_\psi:\{A:0\leq A\leq I\}\longrightarrow[0,1]
\]
such that
\[
v_\psi(I)=1,\qquad
v_\psi(A+B)=v_\psi(A)+v_\psi(B)\quad\text{whenever }A+B\leq I,
\]
and
\[
P_E([\psi])_a=v_\psi(E_a)
\]
for every finite POVM $E$ and every outcome $a$.
\end{proposition}

\begin{proof}
Assume (N) and define $v_\psi(A)=P_{\{A,I-A\}}([\psi])_1$.
A deterministic one-versus-rest merge gives $P_E([\psi])_a=v_\psi(E_a)$, so the value assigned to an effect is independent of the POVM in which it occurs.
The unique distribution on the one-outcome POVM $\{I\}$, followed by the deterministic split to $\{I,0\}$, gives $v_\psi(I)=1$.
For $A+B\leq I$, merge the first two outcomes of $\{A,B,I-A-B\}$ and compare with the separate binary margins to obtain $v_\psi(A+B)=v_\psi(A)+v_\psi(B)$.

Conversely, suppose such a family $v_\psi$ exists.
Finite additivity gives monotonicity: if $0\leq A\leq B\leq I$, then $v_\psi(B)=v_\psi(A)+v_\psi(B-A)\geq v_\psi(A)$.
It also gives $v_\psi(rA)=r\,v_\psi(A)$ for rational $r\in[0,1]$ by splitting $A$ into equal positive parts.
For real $t\in[0,1]$, squeeze $t$ between rational sequences and use monotonicity to obtain $v_\psi(tA)=t\,v_\psi(A)$.
Therefore, for every stochastic $\Lambda$,
\[
\begin{aligned}
P_{\Lambda_*E}([\psi])_b
&=v_\psi\!\left(\sum_a\Lambda(b|a)E_a\right)\\
&=\sum_a\Lambda(b|a)v_\psi(E_a)
 =\bigl(\Lambda P_E([\psi])\bigr)_b,
\end{aligned}
\]
which is (N).
\end{proof}

Inspection of the forward implication shows that only deterministic outcome maps were used. Therefore, on the domain of all finite POVMs, naturality under deterministic coarse-grainings already yields the additive valuation in Proposition~\ref{prop:naturality-additivity}; the converse implication then reconstructs naturality for every stochastic kernel. Restricting (N) from stochastic post-processings to deterministic ones does not weaken the axiom on this domain.

Thus, on finite POVMs, full stochastic naturality is not merely suggestive of effect noncontextuality: it is exactly effect-tuple functionality together with a normalized finitely additive probability assignment on effects.
This is the finite-outcome hypothesis class of Busch's POVM extension of Gleason's theorem \cite{Busch2003} and of the POVM frame-function analysis of Caves, Fuchs, Manne, and Renes \cite{CFMR2004}.
Theorem~\ref{thm:povm-closure} could therefore be recovered from Busch's representation theorem together with Corollary~\ref{cor:all-contexts}; I include the direct binary-margin proof because it is elementary, dimension-uniform, and exposes the exact bridges used.
The point is not that additivity has been eliminated.
It has been repackaged as naturality, while the projective values entering the closure are fixed separately by metric rigidity rather than by the additive assignment itself.

Yang and Fullwood give the closest categorical comparison: their measurement and probability functors are natural under deterministic measurable coarse-grainings, and their converse from natural transformations to quantum states invokes the Busch--Gleason representation \cite{YangFullwood2026}.
For the assignment on all finite POVMs considered here, however, deterministic naturality already extends to stochastic kernels by Proposition~\ref{prop:naturality-additivity}; the deterministic-versus-stochastic distinction is therefore not a weaker-principle escape route.
The categorical formulations can still differ in their choices of objects, morphisms, and functionality assumptions, so I do not identify the axiom systems without a separate translation.
On the restricted domain of projective measurements and their classical mixtures, Wright and Weigert derive the Born rule and density-operator formalism for $d\geq2$ from consistent probability assignments to those outcomes \cite{WrightWeigert2019}; there the consistency condition constrains the projective fragment, whereas here that fragment is already fixed contextwise by Theorems~\ref{thm:pvm-readout} and~\ref{thm:rank-general}.
Accordingly, the statement that noncontextuality arrives as an output rather than an axiom (Remark~\ref{rem:gleason-contrast}) is scoped strictly to projective contexts.
The finite-POVM extension imports effect-noncontextual content deliberately and visibly through (N).

\subsection{Naimark dilation as an alternative closure}

\begin{remark}[Dilation route]\label{rem:naimark}
An alternative conditional closure runs through the arbitrary-rank theorem.
For a finite POVM $E$ define the isometry $V\psi=\sum_aE_a^{1/2}\psi\otimes\ket{a}$ into $\mathcal K=\C^d\otimes\C^m$ (the standard finite-dimensional Naimark dilation \cite{BuschLahtiMittelstaedt1996}) and the projective measurement $Q=\{Q_a\}_{a=1}^m$ with $Q_a=I\otimes\ket{a}\bra{a}$, so that $V^\dagger Q_aV=E_a$.
If the readout of $Q$ on the enlarged projective space satisfies (H1), (H2), and (H3$'$), Theorem~\ref{thm:rank-general} gives $P_Q([\Phi])_a=\braket{\Phi}{Q_a\Phi}$, and the dilation-consistency axiom
\[
P_E([\psi])_a=P_Q([V\psi])_a
\tag{D}
\]
yields $P_E([\psi])_a=\braket{V\psi}{Q_a\,V\psi}=\braket{\psi}{E_a\psi}$.
Naimark's theorem alone supplies the projective representation; it does not identify the probability assignment of the original apparatus with that of the dilation.
Axiom (D) is where the implementation-equivalence content enters.
Calibration on $\mathcal K$ needs no entangled preparations: every unit vector of $\operatorname{range}Q_a=\C^d\otimes\ket{a}$ is a product state, although the embedded state $V\psi$ is generally entangled.
The costs are an ancilla, the projective hypotheses in the enlarged dimension $dm$, and (D) itself, together with either the canonical dilation above or an explicit invariance across equivalent dilations.
The two closures spend their assumptions in different currencies: (N) spends effect-functionality inside $\C^d$, while (D) spends ancillary and cross-dimensional implementation consistency.
Their agreement is a coherence check, not a derivation of one assumption from the other.
\end{remark}

\begin{remark}[The open POVM problem]\label{rem:open-povm}
Section~\ref{sec:povm-closure} does not provide a fixed-POVM analogue of Theorems~\ref{thm:pvm-readout} and~\ref{thm:rank-general}.
It proves the obstruction and displays two ways to add back enough structure.
A genuinely stronger POVM theorem would have to derive, rather than assume, either the binary effect-functionality and spectral-smearing bridges used in Theorem~\ref{thm:povm-closure}, or the dilation consistency (D), from independent operational, compositional, or geometric principles weaker than a generalized probability measure on all effects.
Informational completeness, the norm-$1$ property, and post-processing cleanliness do not by themselves recreate the missing anchors.
No such derivation is claimed here.
\end{remark}

\section{Corollaries and Alternative Routes}\label{sec:corollaries}

\subsection{Raw simplex normalization rigidity}

Suppose $f:[0,1]\to[0,1]$ is continuous and satisfies, for every $d\geq 3$,
\[
\sum_{i=1}^d f(u_i)=1
\qquad\text{for all } u=(u_1,\ldots,u_d)\in\Delta^{d-1}.
\]
Then $f(t)=t$ on $[0,1]$. (The hypothesis for the two values $d=3$ and $d=4$ already suffices.)

Reason: comparing the $d=3$ and $d=4$ vertex equations, $f(1)+2f(0)=1$ and $f(1)+3f(0)=1$, gives $f(0)=0$ and then $f(1)=1$.
For $u,v\geq0$ with $u+v\leq1$, set $w=1-u-v$.
Comparing $(u,v,w,0,\ldots,0)$ with $(u+v,w,0,\ldots,0)$ at fixed $d$ then gives Cauchy additivity $f(u)+f(v)=f(u+v)$ on the unit triangle, and continuity forces linearity, so $f(t)=t$.
The multi-dimension hypothesis is necessary, not decorative: for a single fixed $d$, the affine family
\[
f(t)=(1-da)\,t+a,\qquad 0<a\leq 1/d,
\]
satisfies the normalization identity on $\Delta^{d-1}$ and is not the identity.
The raw componentwise class is therefore rigid across dimensions before any Fisher argument enters, and only across dimensions.

\subsection{Escort-class Born uniqueness}

Let
\[
T_f(u)_i=\frac{f(u_i)}{\sum_j f(u_j)},
\]
with $f$ continuous on $[0,1]$, $C^1$ on $(0,1)$, strictly increasing, and $f(0)=0$.
Then $T_f$ is continuous on the closed simplex, interior-preserving, $C^1$ on $\Delta^{\circ}_{d-1}$, and vertex-fixing; in particular any escort readout of this class satisfies (H3) automatically, since $f(0)=0$ and $f(1)>0$.
Suppose an escort readout for the fixed PVM factors as
\[
P_M^{(f)}([\psi])=T_f\bigl(|\braket{e_1}{\psi}|^2,\ldots,|\braket{e_d}{\psi}|^2\bigr)
\]
and satisfies (H2).
Then $T_f$ is Fisher non-expanding on $\Delta^{\circ}_{d-1}$.
Indeed, if the pointwise inequality failed at some interior tangent pair $(p,v)$, take $p(s)=p+sv$ for sufficiently small $|s|$ and the associated full-support constant-phase curve $\psi(s)=\sum_i\sqrt{p_i(s)}\,\ket{e_i}$.
Along this curve Born saturates $F_{\mathrm{cl}}=F_Q$ by Lemma~\ref{lem:phase-variance}.
Since $T_f$ is $C^1$ and the Fisher form is smooth on the open simplex, strict expansion at $s=0$ would persist on an open parameter interval, contradicting the almost-everywhere condition (H2).
Therefore, for every interior pair $(p,v)$,
\[
g^F_{T_f(p)}\bigl(\dd T_f|_p\cdot v,\dd T_f|_p\cdot v\bigr)
\leq g^F_p(v,v).
\]
Theorem~\ref{thm:simplex-rigidity} then forces $T_f=\id$ on the simplex, so the escort readout collapses to Born.

For $d\geq 3$ this further pins the generator: $T_f=\id$ means $f(u_i)=u_i\sum_j f(u_j)$ on the interior, so any $t,t'\in(0,1)$ with $t+t'<1$ occur as two coordinates of one interior point and give $f(t)/t=f(t')/t'$.
For arbitrary $t,t'\in(0,1)$, choose $0<r<\min\{1-t,1-t'\}$ and compare each of $t,t'$ with $r$; hence $f(t)/t$ is constant and continuity gives $f(t)=ct$ on $[0,1]$ for some $c>0$.
Positive rescaling leaves $T_f$ unchanged.
At $d=2$ the map-level conclusion $T_f=\id$ still holds (the readout is Born) but no longer pins $f$: the nonlinear $f(t)=t(1+t-t^2)$ belongs to the stated generator class because $f'(t)=(1-t)(3t+1)>0$ for $0<t<1$, and it satisfies $T_f=\id$ on $\Delta^{1}$ because $f(t)=t\,S(t)$ with $S(t)=1+t(1-t)$ symmetric under $t\mapsto 1-t$.
The older escort derivations are therefore corollaries of the simplex rigidity theorem rather than the main architecture.

\subsection{Markov-invariance route}

There is also an independent probabilistic route inside the escort class.
If one imposes split-merge invariance under the Campbell--Reginatto congruent-embedding pattern,
\[
f(ks)+f((1-k)s)=f(s),
\]
then $f$ is additive on the unit triangle, and continuity again gives $f(t)=ct$.
This route is logically independent of Theorem~\ref{thm:pvm-readout}.
It does not prove the fixed-PVM readout theorem, but it reaches the same escort linearity by a Markov/coarse-graining axiom rather than by Fisher non-expansion.

\section{Discussion}

Theorems~\ref{thm:pvm-readout} and~\ref{thm:rank-general} do not reconstruct the full measurement postulates of quantum mechanics.
They say that once one fixed projective measurement is operationally identified and calibrated, square-root regularity plus a single metric admissibility condition freeze the rest of its pure-state readout.

Cencov and Campbell characterize the Fisher metric by invariance under Markov morphisms on probability simplices \cite{Cencov1981,Campbell1986}.
The object of uniqueness here is different: not a metric, but the square-root-regular readout maps of a fixed PVM that are compatible with Fisher non-expansion and calibration.

Other routes to the Born rule use different inputs: Gleason's global measure theorem on projectors \cite{Gleason1957}, Zurek's envariance program \cite{Zurek2005}, and the decision-theoretic program developed by Wallace \cite{Wallace2009}.
This paper is narrower in scope: for one fixed calibrated measurement, the conclusion follows from a direct metric rigidity argument.
The structural contrast with Gleason's route is isolated in Remark~\ref{rem:gleason-contrast}: nothing is assumed across measurement contexts, and the conclusion includes $d=2$, where Gleason's theorem fails; the price is per-apparatus inputs of a logically incomparable kind.
Reginatto's program uses generalized Markov mappings on the K\"ahler-geometric side to constrain the quantum metric and Hamiltonian structure, and it sits adjacent to the Markov-style escort corollary above.
The Fisher-geometric tooling here, the square-root chart on the simplex and Fisher information on probability spaces, overlaps with Reginatto's reconstruction \cite{Reginatto2014} and Caticha's entropic dynamics \cite{Caticha2015}; those programs reconstruct quantum geometry or dynamics, while Theorem~\ref{thm:pvm-readout} fixes the readout of one calibrated apparatus.

The Galley--Masanes and Masanes--Galley--M\"uller programs classify or constrain modified measurement postulates under operational and compositional principles such as purification, local tomography, and no-restriction \cite{GalleyMasanes2018,MasanesGalleyMuller2019}.
Theorem~\ref{thm:rank-general} assumes structures those programs try to derive: a fixed Hilbert-space PVM, calibrated basis labels, pure states, square-root regularity (H1), and the admissibility condition (H2).
The two analyses are defined over different objects (pure-state readouts of one fixed measurement here; measurement postulates constrained across systems and composites there), so neither class contains the other.
The concrete overlap is that both exclude escort deformations: the quartic escort ($f(t)=t^2$), a smooth calibrated single-system modification of the Born weights, already fails (H2), quadrupling classical Fisher information at the barycenter of the simplex.

The dimension dependence also matters. Theorem~\ref{thm:pvm-readout} is stated separately for each $d$, and Theorem~\ref{thm:rank-general} separately for each measurement, because there is no scalar generator behind a general readout map.
By contrast, escort corollaries lift across dimensions only because they are controlled by a single scalar function $f:[0,1]\to\R_+$.
Even there, Section~\ref{sec:corollaries} records two distinct thresholds: escort-generator linearity holds within any single fixed dimension $d\geq 3$ and fails at $d=2$, while the raw normalization class is rigid only when the hypothesis spans at least two dimension values.

The metric proofs cover fixed finite dimension, pure states, and projective measurements: rank-1 in Theorem~\ref{thm:pvm-readout} and arbitrary rank in Theorem~\ref{thm:rank-general}.
Those same metric hypotheses do not determine a general POVM readout.
An effect with $\norm{E_i}<1$ admits no preparation that reports its outcome with certainty, and even when every effect has norm one, the certainty subspaces of a non-projective POVM cannot span $\C^d$ (Lemma~\ref{lem:certainty-blocks}); the anchor equalities behind Remarks~\ref{rem:minimal} and~\ref{rem:anchors-general} then acquire slack, and Proposition~\ref{prop:povm-nonrigidity} and Theorem~\ref{thm:povm-boundary} turn that slack into non-Born readouts that survive the metric hypotheses and every occurrence-certainty anchor required by (H3$''$).
Axiom (N) restores the assignment through effect-functionality and finite additivity, and axiom (D) restores it through dilation consistency, but each imports the identification of unsharp effects rather than deriving it.
What identifies an unsharp effect operationally before the Born rule is assumed?
The metric argument does not answer this; mixed states and weakened calibration semantics remain separate open directions.

\section*{Acknowledgments}

I thank Marcel Reginatto for pointing me to the Cencov/Campbell formalism for congruent embeddings by Markov mappings, and for sharing his 2014 derivation of the K\"ahler reconstruction; the framing of the Markov-invariance route in Section~\ref{sec:corollaries} is directly indebted to that program.

\end{document}